\documentclass[10pt,twocolumn,a4paper,conference]{IEEEtran}

\usepackage{cite}	
\usepackage{graphicx} 
\usepackage[latin1]{inputenc} 
\usepackage[T1]{fontenc} 
\usepackage{amsmath,amsfonts,amsbsy,amssymb} 
\usepackage{mathabx} 
\usepackage{amssymb} 
\usepackage{amsmath}
\usepackage{amsthm}	
\usepackage{mathrsfs}
\usepackage[nolist]{acronym} 
\usepackage{tabularx} 
\usepackage{multirow}
\usepackage{wasysym}
\usepackage{float}
\usepackage{color} 

\usepackage{enumitem} 

\newtheorem{lemma}{Lemma}

\newtheorem{proposition}{Proposition}

\begin{document}

\title{Ultra Reliable Communication via Opportunistic ARQ Transmission in Cognitive Networks} 
\author{Mohammad Shehab, Hirley Alves, and Matti Latva-aho\\
	
	\IEEEauthorblockA{
		Centre for Wireless Communications (CWC), University of Oulu, Finland\\
	}
	Email: firstname.lastname@oulu.fi
}
%
\maketitle


\begin{abstract}
This paper presents a novel opportunistic spectrum sharing scheme that applies ARQ protocol to achieve ultra reliability in the finite blocklength regime. A primary user shares its licensed spectrum to a secondary user, where both communicate to the same base station. The base station applies ARQ with the secondary user, which possess a limited number of trials to transmit each packet. We resort to the interweave model in which the secondary user senses the primary user activity and accesses the channel with access probabilities which depend on the primary user arrival rate and the number of available trials. We characterize the secondary user access probabilities and transmit power in order to achieve target error constraints for both users. Furthermore, we analyze the primary user performance in terms of outage probability and delay. The results show that our proposed scheme outperforms the open loop and non-opportunistic scenarios in terms of secondary user transmit power saving and primary user reliability.  
\end{abstract}


\section{Introduction}\label{introduction}
Spectrum sharing has always been an imminent research topic through the current decade. Due to spectrum scarcity, researchers are studying the use of higher frequency bands e.g. millimeter-wave (mm-Wave) to uphold ultra broadband systems in fifth generation networks. Applying cognitive radio schemes represents a promising alternative and parallel solution at the same time \cite{sharingmm}. The design of such networks is expected to support newly introduced technologies such as machine to machine (M2M) communication promoting the Internet of Things (IoT). Quality of service (QoS) constraints are imposed to fulfill very low latency with expected reliability of higher than 99.9$\%$ \cite{paper1,NokiacMTC2016,latency}.  

In order to achieve \textit{Ultra-Reliable Low Latency Communication} (URLLC) in emerging technologies such as industrial automation \cite{latency}, machines communicate using short messages whenever data sizes are reasonably small such as sensor readings or alarm notifications, which is the case in the most of machine type communication (MTC) scenarios \cite{paper1,latency}. Consequently and as a result of the failure of Shannon's model to provide an accurate benchmark for it, finite blocklength communication has been extensively studied recently \cite {paper5,paper1,paper3,eucnc}. For instance, \cite{paper5} characterizes the throughput of delay constrained systems communicating on short packets, while \cite{paper3} defines the maximum achievable rate and throughput of ARQ protocols in the finite blocklength regime.

On the other hand, cognitive radios allow a primary user (PU) to share its licensed channel resources namely spectrum to unlicensed cognitive secondary user (SU). Previous works considered collision scenarios where the SU is allowed to transmit with constraints on the interference temperature affecting the PU \cite{CIT,bedewy,comletter,sharingmm,Makki}. For example, \cite{CIT} proposes a three-dimensional Markov chain model to analyze the SU performance in dynamic spectrum access with interference temperature constraints. In \cite{comletter}, the authors studied the capacity and optimal power allocation in collision scenarios with PU SINR guarantees. Fan et al. introduced a non-cooperative game to maximize the throughput of mm-Wave ultra-dense networks (UDNs) using dynamic spectrum sharing in \cite{sharingmm}. Makki et al. approached finite blocklength spectrum sharing via rate adaption in \cite{Makki} and suggested ARQ protocol as a potential extension to their work. However, non of these endeavors studied applying ARQ and successive interference cancellation (SIC) in the finite blocklength regime.

In this paper, we propose a novel opportunistic transmission framework for a SU in interweave model in the finite blocklength regime. In an interweave model, the SU senses the PU activity and decides whether to transmit or not according to some QoS (outage) guarantees \cite{bedewy}. If the PU is active, the SU accesses the channel with a certain access probability based on the PU and SU arrival rates as characterized in our work. The base station (BS) applies ARQ and SIC on the SU packet. In ARQ protocol, the SU is allowed to retransmit its packet if it receives a NACK feedback from the BS, which means the packet is not successfully decoded. The SU possesses $M$ trails to transmit a single packet, where $M$ depends on the SU arrival rate to retain the SU queue stability. Once the BS decodes the SU packet, it applies SIC to eliminate interference from the PU packet and subsequently, reduces the error outage probability for the PU. Furthermore, we analyze the PU expected delay, which occurs due to the SIC process and the SU retransmissions. The results show that this scheme provides ultra reliability for both the PU and SU while reducing the SU transmit power.

The rest of the paper is organized as follows: in  Section \ref {system model}, we characterize communication at finite blocklength and introduce the system model. Next, Sections \ref{section3} and \ref{section4} include derivations of the SU outage probability and transmit power. After that, we analyze the PU outage probability and delay in Section \ref{PU}. The performance of the proposed scheme is evaluated in Section \ref{results}. Finally, Section \ref{conclusion} concludes the paper.

\section{System layout} \label{system model}
Consider an uplink scenario where the PU and the SU convey short packets with a fixed rate $R$ bits per channel use (bpcu) to a common BS. For finite blocklength transmission, packets are conveyed with error probability $\epsilon \in\left[ 0,1\right]$ given by \cite{paper2,paper3}
\begin{align}\label{e1}
\epsilon(\sigma)=\mathbb{E}_\sigma\left[ Q\left(\frac{n \log_2(1+\sigma )-nR+0.5\log_2n}{\sqrt{nv(\sigma)}} \right)\right] , 
\end{align}  
\begin{align}\label{e2}
v(\sigma)=\left( 1-\frac{1}{\left( 1+\sigma\right) ^{2}}\right) \left( \log_2e\right)^2, 
\end{align} 
where $\sigma$ denotes the SINR, $n$ is the blocklength and $v(\sigma)$ is the channel dispersion.

Both PU and SU have a Bernoulli distributed arrival process with packet arrival probabilities of $\lambda_p$ and $\lambda_s$, respectively as in \cite{system_model}. This implies that although we know the average arrival probability of the PU, we can not predict exactly when the PU will have packets to transmit. Whenever a packet arrives to the PU, the PU transmits it only once and remains silent until the next packet arrival. On the contrary, the SU applies opportunistic transmission. The SU possesses $M$ trials to transmit a single packet where $m\in\left\lbrace 1,2, ... , M \right\rbrace$ denotes the $m^{th}$ trial. According to Loynes' theorem \cite{loynes}, a queue is stable if its average service rate is higher than the average arrival rate. To guarantee SU queue stability at steady state and boost reliability, M varies from one packet to the other such that $M=\lceil \frac{1}{\lambda_s} \rceil$ for $\alpha$ fraction of time and $M=\lfloor \frac{1}{\lambda_s} \rfloor$ for $1-\alpha$ fraction of time. $\alpha$ is selected to be slightly higher than or equal to $\alpha_o$ with
\begin{align}\label{e3}
\alpha_o=\mod\left(\frac{1}{\lambda_s},\lfloor \frac{1}{\lambda_s} \rfloor \right).
\end{align} 

On the other side, the BS applies the ARQ protocol with the SU where the BS calls the SU to retransmit its packet if an error occurs. Given that the SU's packet is not successfully decoded before the $m^{th}$ trial, the SU accesses the channel with probability 1 if the PU is silent and with access probability $q_m \in \left\lbrace q_1, q_2,..., q_M\right\rbrace$ if the PU is transmitting a packet. Here, we assume that the SU is able to sense the PU transmission perfectly as in \cite{CIT}. 

Let $p_p$ and $p_s$ be the transmit powers of the PU and SU, respectively, while $|h_p|$ and $|h_s|$ denote their channel coefficients, respectively as shown in Fig. \ref{system layout}. Both channels are i.i.d quasi-static Rayleigh fading with coherence time of $n$ symbol periods, and the channel coefficients are available at the BS. The noise entries are additive complex Gaussian of unit variance. Define the PU packet outage probability as $\epsilon_p$ and the SU packet outage probability after exhausting its $M$ trials as $\epsilon_s$. 
\begin{figure}[!t] 
	\centering
	\includegraphics[width=1\columnwidth]{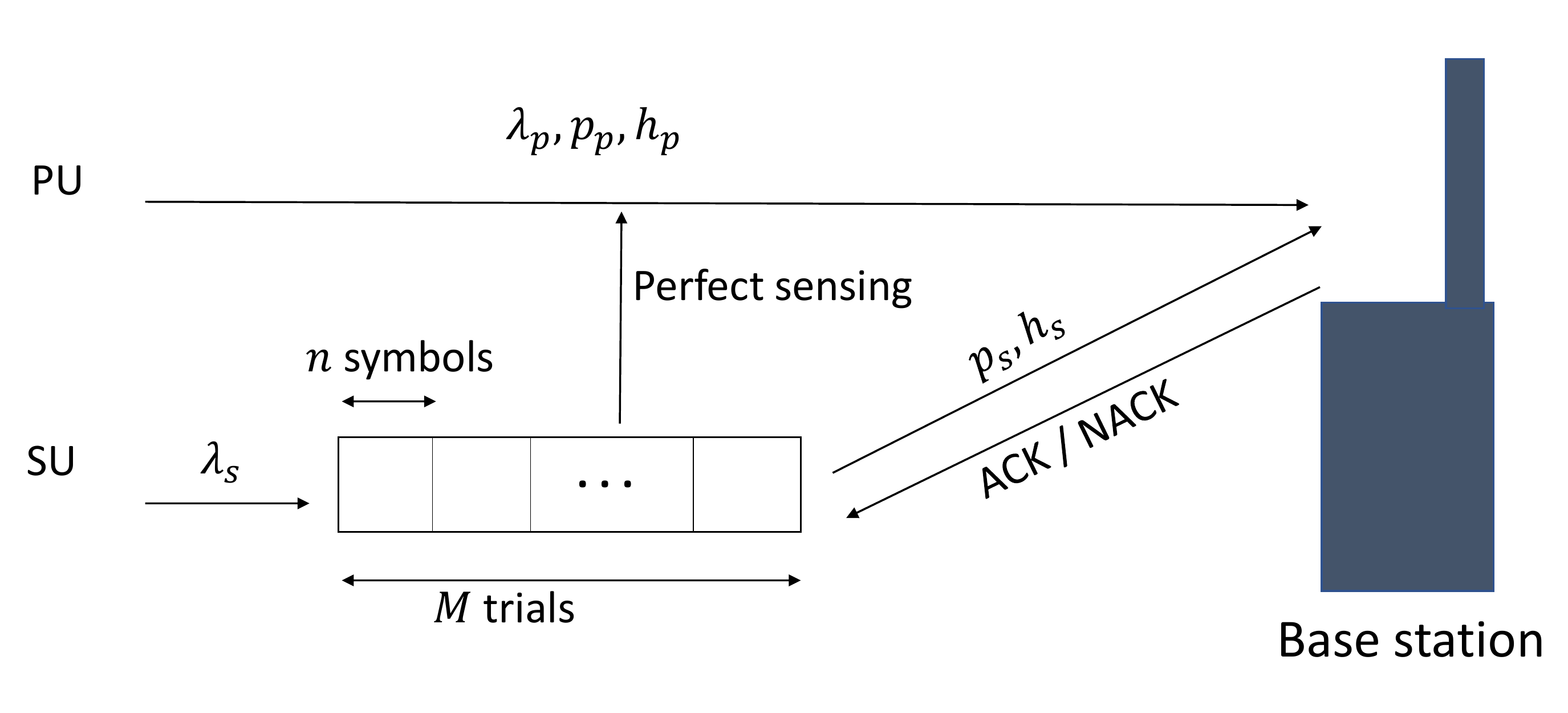}
	\vspace{-6mm}
	\caption{System layout.}
	\label{system layout}
	\vspace{-4mm}
\end{figure}
If both users transmit at the same time, the BS applies SIC on the PU received signal. The BS is able to decode the SU interfering packet successfully with probability $1-\epsilon_s$. In this case, the PU SINR is the same as if the PU transmits while the SU is silent, which is $\sigma_p=p_p|h_p|^2$. This renders an outage probability of $\epsilon_{p_1}=\epsilon(\sigma_p)$. Likewise, the BS fails to decode the SU packet with probability $\epsilon_s$, which renders a PU SINR of $\sigma_{p_i}=\frac{p_p|h_p|^2}{p_s|h_s|^2+1}$. This causes an outage probability of $\epsilon_{p_2}=\epsilon(\sigma_{p_i})$. It is straightforward to infer that $\epsilon_{p_2}\geq\epsilon_{p_1}$ due to the interference in the second case. Leveraging, in case of simultaneous transmission, the PU outage probability is 
\begin{align}\label{e5}
\epsilon_{p}=(1-\epsilon_s) \epsilon_{p_1}+\epsilon_s \epsilon_{p_2}. 
\end{align}  

At the $m^{th}$ trial, the SU has an outage probability of $\epsilon_{s_1}=\epsilon(\sigma_s)$ if the PU is silent and $\epsilon_{s_2}=\epsilon(\sigma_{s_i})$ in case of concurrent transmission, where $\sigma_{s}=p_s|h_s|^2$ and $\sigma_{s_i}=\frac{p_s|h_s|^2}{p_p|h_p|^2+1}$. The targets of our analysis are as follows:

\begin{enumerate}
	\item Compute the PU and SU outage probabilities.
	
	\item Optimize the SU access probabilities to minimize the SU transmit power subject to outage constraints.	
	
	\item Assess the performance of opportunistic ARQ transmission in terms of SU power saving and PU reliability. 
	
\end{enumerate}

\section{Secondary user Outage Probability} \label{section3}
We start by deriving the outage expression for the SU in closed form. At a certain instant, suppose that a packet arrives at the SU's buffer; the SU user starts to sense the channel for PU transmission and admits opportunistic transmission with ARQ as detailed in Section \ref{system model}. Define the probability of reaching the $m^{th}$ trial as $P_{SU} (m)$. The probability tree of success and failure at the $m^{th}$ is projected in Fig. \ref{tree}. 

\begin{lemma} \label{EEEV}
The SU's cumulative outage probability $\epsilon_s$ can be formulated as
\begin{align}\label{e19}
\epsilon_s=P_{SU}(M+1)= \prod_{i=1}^{M} \theta_i=\prod_{i=1}^{M} \beta q_i + T ,
\end{align}
where $\beta=\lambda_p (\epsilon_{s_2}-1)$, and $T=\lambda_p (1-\epsilon_{s_1})+\epsilon_{s_1}$.
\end{lemma}
\begin{figure}[!t] 
	\centering
	\includegraphics[width=0.7\columnwidth]{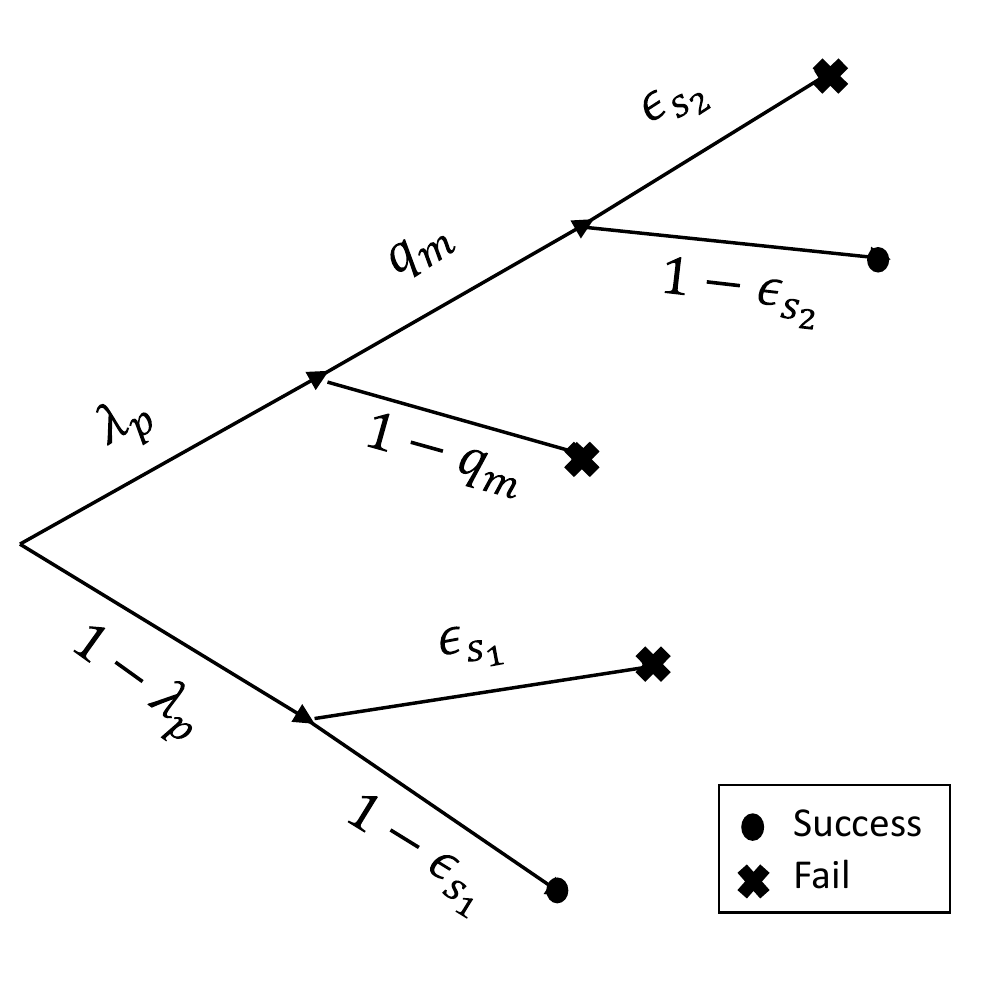}
	\vspace{-4mm}
	\caption{Probability tree of the $m^{th} $ trial.}
	\label{tree}
	\vspace{-4mm}
\end{figure}
\begin{proof} 
At the $m^{th}$ trial and according to Fig. \ref{tree}, we define the following probabilities; the probability that the SU transmits at the $m^{th}$ trial is given by
\begin{align}\label{e11}
P_{SU}(\mathrm{Tx} \ | \ m)=\lambda_p q_m+(1-\lambda_p)=1-\lambda_p(1-q_m),
\end{align} 
while the probability that it remains silent is
\begin{align}\label{e12}
P_{SU}(\mathrm{no \ Tx} \ | \ m)=\lambda_p(1-q_m).
\end{align} 
Thus, the probability of success at the $m^{th}$ trial is
\begin{align}\label{e13}
P_{SU}(\mathrm{suc} \ | \ m )= \lambda_p q_m (1-\epsilon_{s_2}) + (1-\lambda_p) (1-\epsilon_{s_1}).
\end{align} 
For the $m^{th}$ time slot, we obtain the fail probability as
\begin{align}\label{e14}
P_{SU}(\mathrm{fail \ | \ m})&= 1-P_{SU}(\mathrm{suc} \ | \ m ) \notag \\
=\lambda_p q_m &\epsilon_{s_2}+\lambda_p (1-q_m)+(1-\lambda_p)\epsilon_{s_1},
\end{align}
which for $m=1$ corresponds to the probability of reaching the second trial. Similarly, we obtain the probability of reaching the third trial as 
\begin{flalign}\label{e16}
P_{SU}(3)=&P_{SU}(\mathrm{fail \ | \ 1}) \cdot P_{SU}(\mathrm{fail \ | \ 2}) \notag \\
=&\left[ \lambda_p q_1 \epsilon_{s_2}+\lambda_p (1-q_1)+(1-\lambda_p)\epsilon_{s_1}\right] \cdot \notag  \\
& \left[\lambda_p q_2 \epsilon_{s_2}+\lambda_p (1-q_2)+(1-\lambda_p)\epsilon_{s_1}\right].
\end{flalign}
Following the above pattern, we attain the probability of reaching the $m^{th}$ trial as
\begin{align}\label{e18}
P_{SU}(m)&=\prod_{i=1}^{m-1}\lambda_p q_i \epsilon_{s_2}+\lambda_p (1-q_i)+(1-\lambda_p)\epsilon_{s_1} \notag \\
&=\prod_{i=1}^{m-1} \beta q_i + T.
\end{align}
The SU's cumulative outage probability $\epsilon_s$ can be projected as the virtual probability of reaching the $(M+1)^{th}$ trial after the $M^{th}$ trial fails. That is, $\epsilon_s=P_{SU}(M+1)$ which according to (\ref{e18}) leads to (\ref{e19}).
\end{proof}
\begin{proposition} \label{p1}
The SU outage probability is lower bounded by $\epsilon_{s_l}=\left(\beta+T\right)^M$, and upper bounded by almost surely $ \epsilon_{s_u}\stackrel{a.s.}{=}T^M$.
\end{proposition}
\begin{proof}
The lowest possible SU outage probability $\epsilon_{s_l}$ occurs when the SU always transmits regardless of the state of the PU. Substituting \{$q_m=1, \  \forall m$\} in (\ref{e19}), we obtain
\begin{align}\label{e20}
\epsilon_{s_l}\!=\prod_{i=1}^{M} \beta + T=\left(\beta+T\right) ^M
\end{align}
Likewise, the worst case outage occurs when the SU never transmits whenever the PU is active; that is \{$q_m=0, \  \forall m$\} in (\ref{e19}), which yields 
\begin{align}\label{e21}
\epsilon_{s_u}\stackrel{a.s.}{=}\prod_{i=1}^{M} T =T^M.
\end{align}
\end{proof}
\vspace{-2mm}
Here, $\epsilon_{s_l}$ and $\epsilon_{s_u}$ are the lower and upper bounds, respectively for the SU outage probability and $\epsilon_{s}$ lies in the interval $\left[ \epsilon_{s_l},\epsilon_{s_u} \right]$, where the value of $\epsilon_{s}$ is determined according to the access probabilities $\left\lbrace q_1, q_2,..., q_M\right\rbrace$.
\vspace{-1mm}

\section{Secondary user transmit power} \label{section4}
In this section, we derive the SU transmit power required to achieve a target SU outage probability $\epsilon_{st}$ in the open loop, non-opportunistic, and opportunistic schemes. First, we start by the open loop scenario.
\subsection{Open Loop Scenario}
To achieve a transmission rate $R$ at a target SU error probability $\epsilon_{s_t}$ with single transmission (open loop), the amount of allocated power $p_s^*(\sigma)$ according to (\ref{e1}) is the root of
\begin{align}\label{e22}
\mathbb{E}_\sigma\left[Q\left(\frac{n \log_2(1+\sigma )-nR+0.5\log_2n}{\sqrt{nv(\sigma)}} \right)\right] =\epsilon_{s_t}, 
\end{align}
which can be obtained using Matlab root-finding functions, e.g., fzero or plotting when setting $\sigma=\sigma_{s}=p_s^*|h_s|^2$ in a similar way to \cite{comletter}. Taking into consideration the interference that that occurs when the PU has a packet arrival with probability $\lambda_p$, the SU SINR becomes $\sigma=\sigma_{s_i}=\frac{p_s^*|h_s|^2}{p_p|h_p|^2+1}$. Thus, the open loop SU transmission power is
\begin{align}\label{e22'}
p_{s_{ol}}=\lambda_p p_s^*(\sigma_{s_i})+(1-\lambda_p)p_s^*(\sigma_{s}), 
\end{align}
where $p_s^*(\sigma_{s})$ and $p_s^*(\sigma_{s_i})$ are obtained by solving (\ref{e22}).
\subsection{Non-opportunistic Scenario}
In non-opportunistic transmission, the SU keeps retransmitting its message with probability 1 till it is successfully decoded regardless of the state of the PU. In this case the SU needs $M_{no}$ trials to convey its packet, which can be obtained from the lower bound equation in (\ref{e20}) as
\begin{align}\label{e22''}
M_{no}=\frac{\log \epsilon_{s_t} }{\beta+T}=\frac{\log \epsilon_{s_t}}{\log\left[  (1-\lambda_p)\epsilon_{s_1}+\lambda_p\epsilon_{s_2}\right] }, 
\end{align} 
Thus, the transmit power is given by
\begin{align}\label{e23}
p_{s_{no}}=\frac{\log \epsilon_{s_t}}{\log\left[  (1-\lambda_p)\epsilon_{s_1}+\lambda_p\epsilon_{s_2}\right] }p_s. 
\end{align} 

\subsection{Opportunistic Scenario}
Back to the transmission process, exploiting (\ref{e11}) and (\ref{e18}), we can say that a transmission occurs at the $m^{th}$ trial with probability \vspace{-1.5mm}
\begin{align}\label{e24}
P_{SU}(\mathrm{m^{th} \ Tx})&=P_{SU}(\mathrm{Tx} \ | \ m)\cdot P_{SU}(m)   \notag \\ 
&= \left[ \lambda_p q_m+(1-\lambda_p) \right]\prod_{i=1}^{m-1} \theta_i.
\end{align} \vspace{-0.25mm}
Thus, the SU transmit power needed to deliver one packet $p_{s_{op}}$ can be formulated as the sum of transmission probabilities $\phi$ times the power per transmission $p_s$. That is \vspace{-1.5mm}
\begin{align}\label{e25}
p_{s_{op}}=\phi \cdot  p_s, 
\end{align} \vspace{-2mm}
with \vspace{-2mm}
\begin{align}\label{e26}
\phi=\underbrace{\lambda_p q_1+(1-\lambda_p)}_{\text{$1^{\mathrm{st}}$ Tx}} + \sum_{k=2}^{M}\left[ \lambda_p q_k+(1-\lambda_p) \right] \prod_{i=1}^{k-1} \theta_i.
\end{align} \vspace{-0.5mm}
Since $p_s$ is constant in (\ref{e26}), it is clear that the power consumed per packet transmission solely depends on the objective function $\phi$. 

Reversing the problem, we aim at computing the necessary access probabilities $q_1, q_2, ... q_M$ to satisfy a certain target SU outage probability $\epsilon_s=\epsilon_{s_{t}}$ such that $\epsilon_{s_{t}}\geq \epsilon_{s_l}$. A feasible solution can be obtained by setting $q=q_1= q_2= ...= q_M$, which implies equal access probabilities for all trials. From (\ref{e19}), we have \vspace{-5mm}
\begin{align}\label{e28}
\epsilon_s=\prod_{i=1}^{M} \theta_i=\prod_{i=1}^{M} \theta=\theta^M=\epsilon_{s_{t}},
\end{align}
where $\theta=\beta q+T$. The solution of (\ref{e28}) gives 
\begin{align}\label{e29}
q=\frac{\epsilon_{s_{t}}^{\frac{1}{M}}-T}{\beta}.
\end{align}
Indeed (\ref{e29}) represents a feasible solution that satisfies the outage constraint. The consumed power per packet can be obtained by substituting the resultant $q$ in (\ref{e26}), then (\ref{e25}). However, this solution seems to be sub-optimal since there is no evidence that it minimizes the objective function $\phi$ and hence, the consumed power per packet in (\ref{e25}).

To obtain the optimum access probabilities, we cast the following optimization problem to minimize the objective function $\phi$ with both PU and SU outage constraints \vspace{-1mm}
\begin{align}\label{op1}
\min_{q_1,..,q_M} \ &\phi \ ,  \\ 
s.t \  \ \ &\epsilon_s \leq \epsilon_{s_{t}} \notag \\
&\epsilon_p \leq \epsilon_{p_{t}} \notag \\ 
&0\preceq\left\lbrace q_1,..,q_M\right\rbrace  \preceq 1. \notag
\end{align}
Taking a closer look on the above problem, we recognize that neither $\phi$ nor $\epsilon_s$ has a positive definite Hessian matrix an hence, neither is convex. Consequently, this problem can not be solved using Karush-Kuhn-Tucker (KKT) conditions \cite{Boyd}. Nevertheless, this problem can be easily solved by a numerical search with the same approach implemented in \cite{bedewy}, since the access probabilities are bounded between 0 and 1. Moreover, the numerical search can be further simplified bearing that the optimum solution must be such that $q_1\leq .. \leq q_M$, which can be proven, and the mathematical proof is omitted due to lack of space. 

\vspace{-0.8mm}
\section{Primary user Performance analysis} \label{PU}
Herein, we analyze the PU performance when the SU packet inter-arrival time is constant, which is the worst case scenario for the PU as the SU always has a new packet to transmit once it exhausts its $M$ trials for the previous packet. \vspace{-1mm}
\subsection{Primary User Outage Probability}
\begin{proposition}
The minimization of the objective function $\phi$ not only minimizes the SU transmit power, but also minimizes the PU overall outage probability $\epsilon_p$.
\end{proposition}
\begin{proof}
First, we characterize the PU outage probability. Considering that (\ref{e26}) represents the SU transmission probabilities and that a PU packet can arrive at any SU trial with probability $\frac{1}{M}$, we can infer the probability of simultaneous transmission in terms of $\phi$ as \vspace{-1.5mm}
\begin{align}\label{e30}
P(\mathrm{Simultaneous \ Tx})= \frac{\phi}{M}.
\end{align}
From (\ref{e5}), the PU outage probability can be written as
\begin{align}\label{e31}
\epsilon_p=\frac{\phi}{M}\left[ (1-\epsilon_s) \epsilon_{p_1}+\epsilon_s \epsilon_{p_2}\right] + \left(1- \frac{\phi}{M}\right) \epsilon_{p_1},
\end{align}
and the first derivative of $\epsilon_p$ w.r.t $\phi$ is 
\begin{align}\label{e32}
\frac{\partial \epsilon_p}{\partial \phi}\!=\frac{1}{M}\left( \epsilon_s \epsilon_{p_{2}}\!+\epsilon_{p_{1}}\!- \epsilon_s \epsilon_{p_{1}}\! -\epsilon_{p_{1}}\right) =\!\frac{\epsilon_s}{M}\left(\epsilon_{p_{2}}\!-\epsilon_{p_{1}} \right), 
\end{align}
which is strictly positive since $\epsilon_{p_{2}} > \epsilon_{p_{1}}$.
\end{proof}

\subsection{Primary User Delay}
When decoding the PU packet, an extra delay occurs when the PU suffers from interference from the SU. This is because the decoder must wait for the SU packet to be decoded first to perform interference cancellation which may take up to $M$ time slots; then it can decode the PU packet without interference. To analyze the delay $\delta$, first we assume that once the SU packet has been received with no error, the BS requires a virtual zero processing time to decode the SU packet and perform SIC. Thus, the only delay occurs due to the decoder waiting time for the SU retransmissions. A PU packet can arrive at any SU trial with probability $\frac{1}{M}$. Therefore, at the $m^{th}$ time slot, PU suffers a virtual zero delay in the following cases:
\begin{enumerate}
\item The SU is not present (did not reach the $m^{th}$ trial).
\item It is present and decides not to transmit or transmits successfully.
\item The SU is at its $M^{th}$ trial. 
\end{enumerate}
Thus, the virtual zero latency probability can be formulated by (\ref{e33}) on the the top of the next page and the above 3 cases are highlighted below each term. Following a similar procedure, we attain the probability of integer ($l:0<l<M$) packet duration latency as (\ref{e34}). The PU expected delay is bounded by \vspace{-2 mm}
	\begin{align}\label{e38}
	\begin{split} 
	\mathbb{E}_l\left[\delta \right]=\sum_{l=0}^{M-1}l.P(\delta=l) .
	\vspace{-2 mm}
	\end{split}
	\end{align}
\begin{figure*}[!t]
	\begin{align}\label{e33}
	\begin{split} 
	P(\delta=0)=\frac{1}{M} \left[ \sum_{m=1}^{M-1}\left[ \left( 1-\prod_{i=0}^{m-1} \theta_i\right)_{\text{1}} +\left(\prod_{i=0}^{m-1} \theta_i \right)_{\text{2}}  \left(q_m(1-\epsilon_{s_2})+(1-q_m) \right) \right]+1_3 \right].
	\end{split}
	\end{align}
    \hrule

	\begin{align}\label{e34}
	\begin{split} 
	P(\delta=l)=\frac{1}{M}\left[ \sum_{m=1}^{M-\delta-1}\left[ q_m\epsilon_{s_2}\left( \lambda_p q_{m+l} (1-\epsilon_{s_2}) + (1-\lambda_p)  (1-\epsilon_{s_1})\right)\prod_{i=1,i\neq m}^{m+l-1} \theta_i \right]  +q_{M-\delta}\epsilon_{s_2}\prod_{i=1,i\neq M-\delta}^{M-\delta+l-1} \theta_i \right]. 
	\end{split}
	\end{align}
	\hrule
	\vspace{-1.5mm}
\end{figure*}

\section{Performance Evaluation} \label{results}
\begin{figure}[!t] 
	\centering
	\includegraphics[width=0.94\columnwidth]{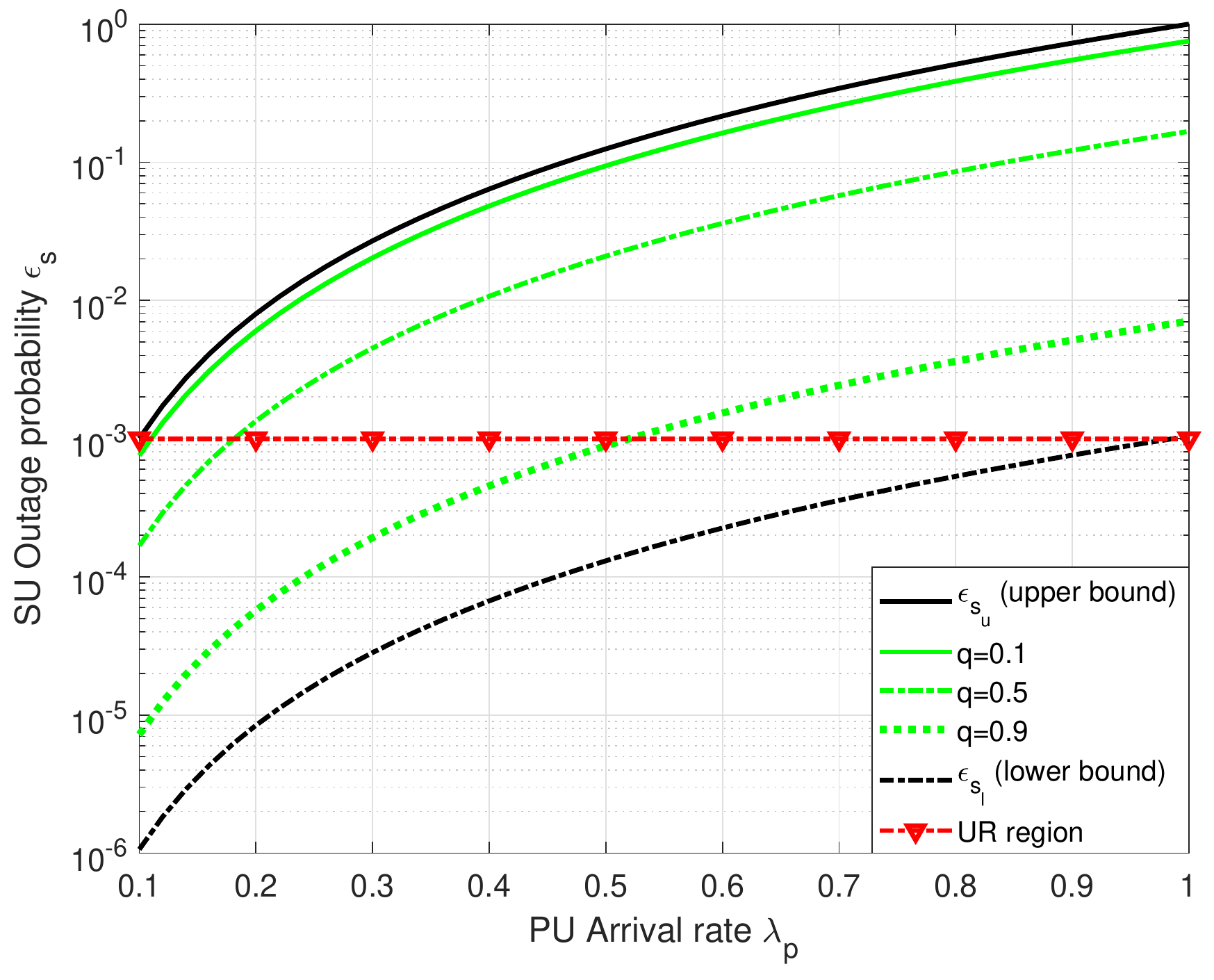}
	\vspace{-2mm}
	\caption{SU outage probability $\epsilon_s$ vs PU arrival rate $\lambda_p$ for $n=500$, $R=0.25$, $M=3$, $p_p=30$ dB, $p_s=32$ dB.}
	\label{su_outage}
	\vspace{-2mm}
\end{figure}

This section includes different plots to elucidate the performance of ARQ opportunistic transmission with SIC in terms of SU outage, SU power allocation, and PU reliability. The following parameters are fixed for all plots: $R=0.25$ bpcu, $p_p=30 $ dB, $p_s=32 $ dB, $n=500$, and $\lambda_s=\frac{1}{3}$ ($M=3$). In Fig. \ref{su_outage}, we plot the SU outage probability $\epsilon_s$ as a function of the PU arrival rate $\lambda_p$ according to (\ref{e19}). In this simulation, the SU access probabilities $q_m$ are constant for all trials as detailed in the figure. The plot shows the SU outage bounds obtained from Proposition \ref{p1}. It is clear that the SU outage probability increases when the PU has higher arrival rate as well as when the SU access probability $q$ is low. Thus, to achieve ultra-reliability for the SU when the PU has high arrival rates, it is essential to rise the SU access probabilities.

In Fig. \ref{Power allocation}, we compare the SU power allocation in opportunistic transmission to non-opportunistic transmission and open loop (one shot) scenarios. The target outage constraints are set as $\epsilon_{s_t}=\epsilon_{p_t}=10^{-3}$ for the next 3 figures. Fig. \ref{Power allocation} shows the power allocation necessary to achieve the target SU outage probability in each scenario for different PU arrival rates $\lambda_p$. The figure depicts that the opportunistic scheme with equal or optimum power allocation requires significantly less power when compared to non-opportunistic transmission, while the open loop  setup is the worst case. The amount of power saving increases when the PU queue becomes more congested, where it reaches more than 3 dB (half power) for $\lambda_p=0.9$ with respect to non-opportunistic transmission and up to 20 dB with reference to the open loop scenario. It is also obvious that more power is consumed when the PU arrival rate $\lambda_p$ increases. We also notice that the power gap between the optimum and equal power allocation is very tight and hence, the equal power allocation strategy highly approaches power optimality with low complexity.

Another observation worth mentioning is that although the transmission rate is low ($R=0.25$) bpcu, we had to transmit with high SNR to achieve ultra-reliability as envisaged by \cite{Dosti}. Bearing in mind that transmit power is normalized to a unit power of noise, transmission with 32 dB in case of opportunistic transmission is considered to be practically plausible specially when compared to 52 dB in the single transmission scheme or 35 dB in the non-opportunistic scenario. The fact that the SU has more chances to transmit when the PU is silent leads to the minimization of the interference for both users. This in turn reduces the power needed by the SU to transmit its packet and the sacrifice of reliability for the PU when the SU is present.

\begin{figure}[!t] 
	\centering
	\includegraphics[width=0.93\columnwidth]{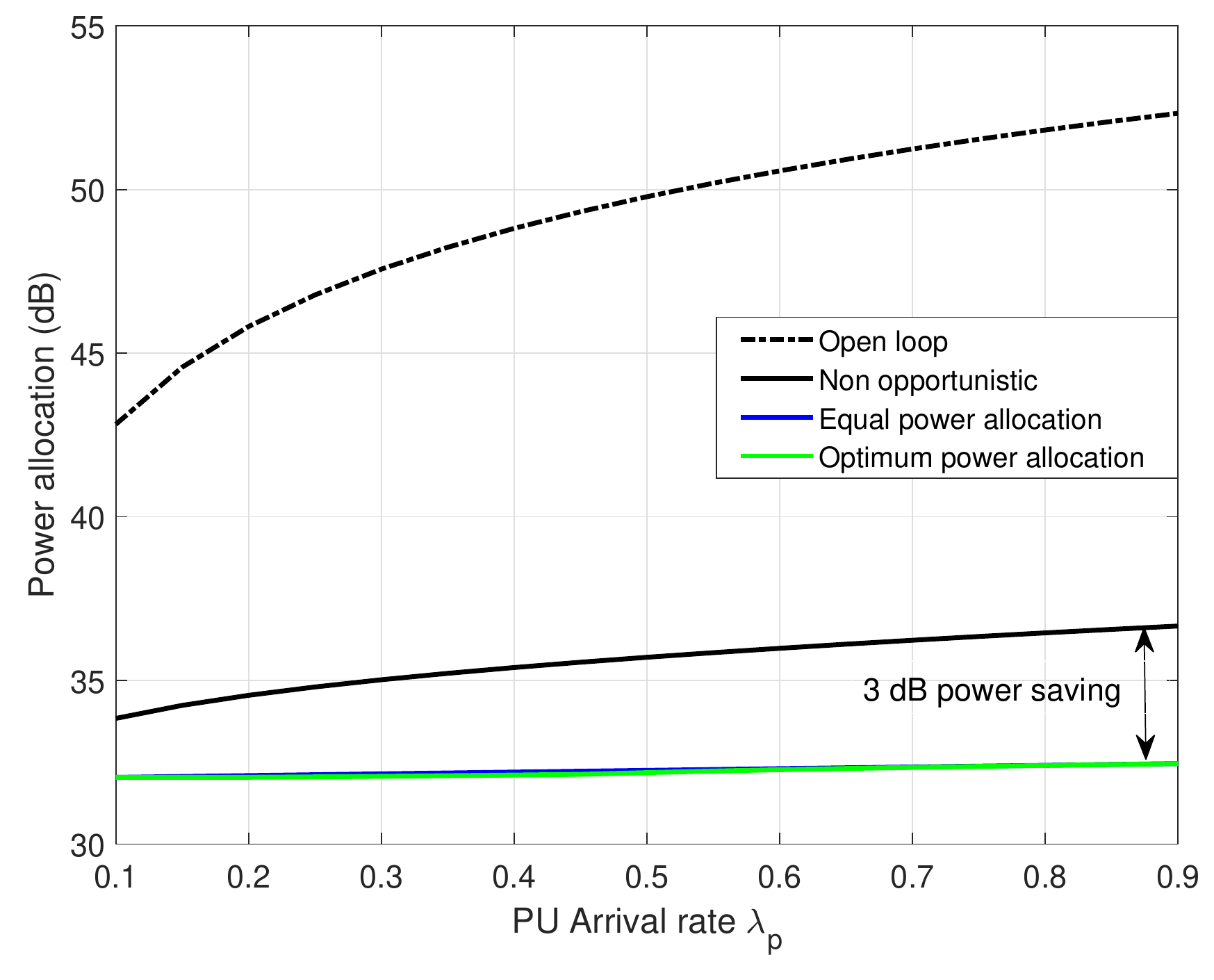}
	\vspace{-2mm}
	\caption{Power allocation vs PU arrival rate $\lambda_p$ for $\epsilon_{s_t}=\epsilon_{p_t}=10^{-3}$, $R=0.25$, $p_p=30 $ dB, $p_s=32 $ dB, $n=500$, and $M=3$.}
	\label{Power allocation}
	\vspace{-2mm}
\end{figure}

Fig. \ref{ep} depicts the PU error outage probability with and without the presence of SU for different PU arrival rates. The figure shows that the PU outage probability rises from $1.5 \times 10^{-4}$ to $2 \times 10^{-4}$. It is obvious that applying the ARQ opportunistic scheme guarantees a very limited declination in the PU reliability in terms of error while serving the SU. Thus, we attain URC for both PU and SU at the same time.
\begin{figure}[!t] 
	\centering
	\includegraphics[width=0.93\columnwidth]{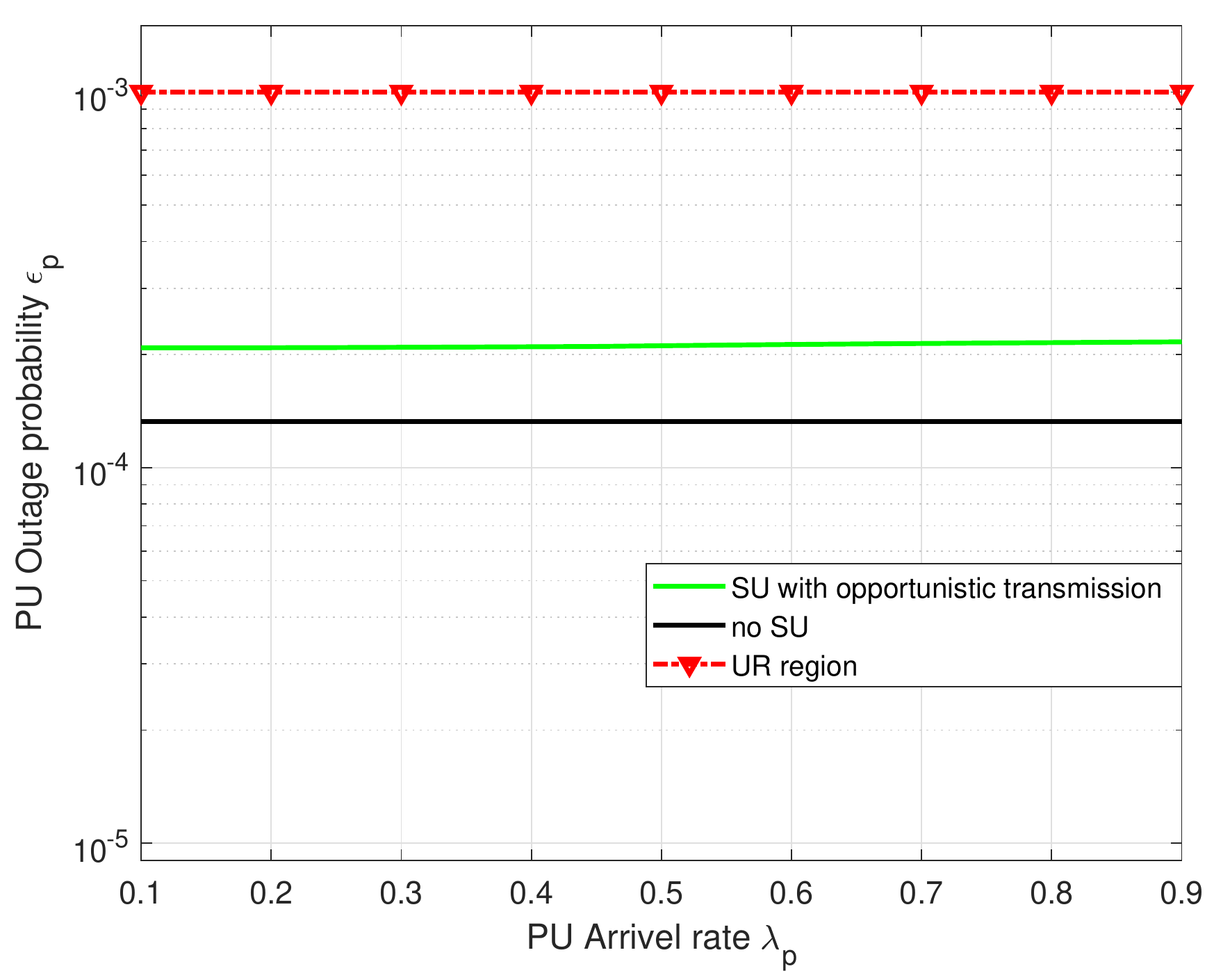}
	\vspace{-2mm}
	\caption{PU error probability $\epsilon_p$ vs PU arrival rate $\lambda_p$ for $\epsilon_{s_t}=\epsilon_{p_t}=10^{-3}$, $R=0.25$, $p_p=30 $ dB, $p_s=32 $ dB, $n=500$, and $M=3$.}
	\label{ep}
	\vspace{-2mm}
\end{figure}

\begin{figure}[!t] 
	\centering
	\includegraphics[width=0.93\columnwidth]{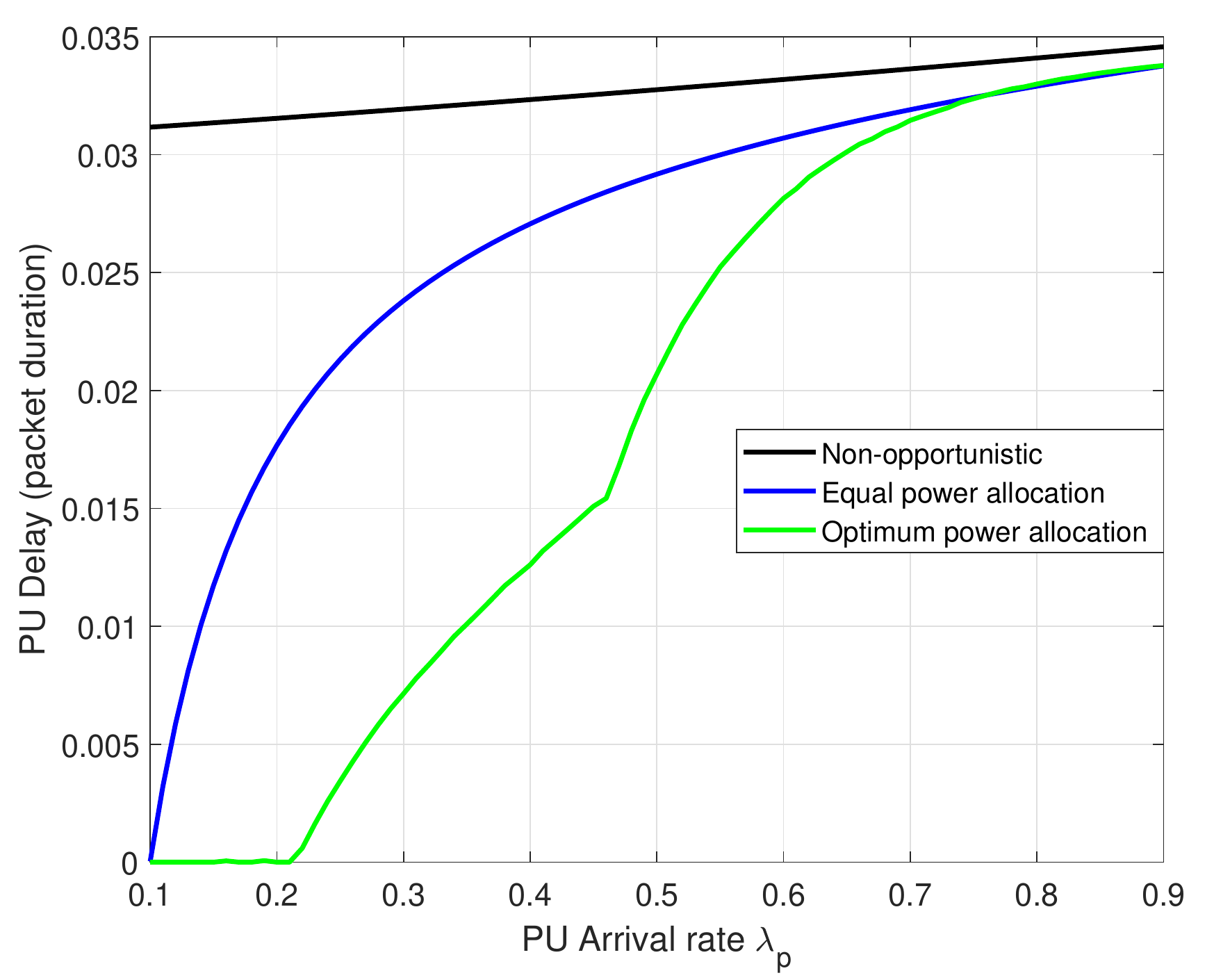}
	\vspace{-2mm}
	\caption{PU average delay $\epsilon_p$ vs PU arrival rate $\lambda_p$ for $\epsilon_{s_t}=\epsilon_{p_t}=10^{-3}$, $R=0.1$, $p_p=30 $ dB, $p_s=32 $ dB, $n=500$, and $M=3$.}
	\label{delay}
	\vspace{-2mm}
\end{figure}

Finally, we plot the PU average delay obtained from (\ref{e38}) as a function of PU arrival rate. The figure shows that the PU average delay for non-opportunistic scheme is higher than when applying the opportunistic transmission whether with equal or optimum power allocation. Unlike, the results in Fig. \ref{Power allocation}, the optimum power allocation strategy renders a considerable reduction in PU delay when compared to equal power allocation. For example, when $\lambda_p=0.5$, the delay is only 2 \% of packet duration. In all cases, the PU delay rises when the PU queue becomes more congested.

\section{Conclusion} \label{conclusion}
In this work, we introduced an efficient spectrum sharing model to achieve URC for short packets. The proposed model relies on ARQ opportunistic transmission of the SU after sensing the PU activity in an interweave scenario. We derived the SU outage expressions and defined the upper and lower bounds for this outage as a function of PU arrival rate. We characterized the optimum access probability and analyzed the average delay of the PU packets. The results showed that opportunistic transmission requires significantly lower power when compared to non-opportunistic and open loop scenarios. Moreover, the equal access probabilities highly approaches optimality in terms of power saving. However, we observe that applying the proposed opportunistic scheme with optimum power allocation strategy significantly reduces the PU delay. The proposed scheme serves the SU with a negligible sacrifice in the PU reliability.

\section*{Acknowledgments}
This work is partially supported by Aka Project SAFE (Grant no. 303532), and by Finnish Funding Agency for Technology and Innovation (Tekes), Bittium Wireless, Keysight Technologies Finland, Kyynel, MediaTek Wireless, and Nokia Solutions and Networks.

\bibliographystyle{IEEEtran}
\bibliography{di}
\end{document}